%% file: Unstable_OU.tex
\newif\ifreport\reportfalse
\documentclass[conference,10pt]{IEEEtran}

\usepackage{setspace}
\usepackage{cite}
\usepackage{graphicx}
\usepackage{caption}
\usepackage{subcaption}
\usepackage[cmex10]{amsmath}
\usepackage{mathtools}
\usepackage{amsthm}
\usepackage{amsfonts}
\usepackage[font=small,skip=0pt]{caption}
\usepackage{amssymb}
\usepackage{algorithm}
\usepackage{algorithmic}
\usepackage{bm,xstring}
\usepackage{fixltx2e}
\usepackage{tikz}
\usetikzlibrary{decorations.pathreplacing}

\usepackage{color}

\newcommand{\ignore}[1]{}

\newtheorem{lemma}{Lemma}

\newtheorem{theorem}{Theorem}

\theoremstyle{definition}

\begin{document}

\title{Samples through Queues: Remote Estimation and Age of Information}
\title{Sampling for Remote Estimation through Queues}
\title{Performance Bounds for Sampling and Remote Estimation of Gauss-Markov Processes over a Noisy Channel with Random Delay} 

\IEEEoverridecommandlockouts
\author{Tasmeen Zaman Ornee and Yin Sun\\ 
Dept. of ECE,  Auburn University, Auburn, AL\\
%
\thanks{This work was supported in part by NSF grant CCF-1813050, ONR grant N00014-17-1-2417, and ARO grant W911NF-21-1-0244}
}
\maketitle
\input{abstract}
\input{sec_intro}
\input{sec_related_work}

\input{model}

\input{main_results}

\input{numerical_results}
 \input{sec_conclusion}

\bibliographystyle{IEEEtran}
\bibliography{ref,ref1,ref_2,sueh}
\appendices
\input{sec_appendix}

\end{document}

%% file: abstract.tex
\begin{abstract}
In this study, we generalize a problem of sampling a scalar Gauss Markov Process, namely, the Ornstein-Uhlenbeck (OU) process, where the samples are sent to a remote estimator and the estimator makes a causal estimate of the observed \emph{real-time} signal. In recent years, the problem is solved for stable OU processes. We present solutions for the optimal sampling policy that exhibits a smaller estimation error for both stable and unstable cases of the OU process along with a special case when the OU process turns to a Wiener process. The obtained optimal sampling policy is a threshold policy. However, the thresholds are different for all three cases. Later, we consider additional noise with the sample when the sampling decision is made beforehand. The estimator utilizes noisy samples to make an estimate of the current signal value. The mean-square error ($\mathsf{mse}$) is changed from previous due to noise and the additional term in the $\mathsf{mse}$ is solved which provides performance upper bound and room for a pursuing further investigation on this problem to find an optimal sampling strategy that minimizes the estimation error when the observed samples are noisy. Numerical results show performance degradation caused by the additive noise.
\end{abstract}


%% file: sec_intro.tex
\section{Introduction}

The problem of sampling an Ornstein-Uhlenbeck (OU) process is recently addressed in \cite{Ornee_TON} and another problem of sampling a Wiener process in \cite{Wiener_TIT}. However, the optimal sampling policy provided in \cite{Ornee_TON} is only for the stable scenario. In practice, real-time applications of OU processes consider both stable and unstable cases \cite{app_un}. Therefore, a sampling problem that considers only the stable scenario is insufficient for practical and more dynamical systems, and a generalization of this problem that considers both stable and unstable cases is necessary. 

Moreover, a real-time system often consists of noise along with the signal process. Therefore, the analysis based on noisy observation of samples to minimize signal estimation error is practically much more important in real-time networked control and communication systems. In this paper, we generalize a sampling problem of a scalar Gauss-Markov process, named the OU process by considering both stable and unstable scenarios. Later on, we consider noisy samples of OU process and compute the $\mathsf{mse}$ from which we establish estimation performance bounds of $\mathsf{mse}$. The optimal sampling policy for noisy samples is not provided in this work but will be considered in our future study.

The OU process is defined as the solution to the following stochastic differential equation (SDE) \cite{PhysRev.36.823,Doob1942}
\begin{align}\label{eq_SDE}
dX_t =\theta (\mu-X_t) dt + \sigma dW_t, 
\end{align}
where $\mu$,  $\theta$,  and $\sigma >0$ are parameters and $W_{t}$ represents a Wiener process. In case of stable OU process, $\theta > 0$ \cite{Ornee_TON}. In \eqref{eq_SDE}, if $\theta \to 0$, and $\sigma =1$, $X_t$ reduces to a Wiener process. If $\theta < 0$, then $X_t$ becomes an unstable OU process. Examples and properties of OU processes are explained in \cite{Ornee_TON}.

First, we aim to find an optimal sampling strategy that minimizes the $\mathsf{mse}$. The samples of the OU process pass through a channel in first-come, first-serve (FCFS) strategy. A remotely located estimator utilizes these causally received samples to make an estimate $\hat{X_t}$ of $X_t$. We obtain lower bound of $\mathsf{mse}$ in the absence of any additional noise in the system. Second, our goal is to find the expression of $\mathsf{mse}$ with the presence of noise in the system. This analysis provides an upper bound of $\mathsf{mse}$ when the estimator receives noisy samples. We summarize the contributions of this paper as follows:

\begin{itemize}

\item The optimal sampling problem in the absence of noise is formulated and the solved optimal sampling policy is a threshold policy on \emph{instantaneous estimation error}. The structure of the thresholds $v(\beta)$ of a parameter $\beta$ are different for the three cases: $\theta>0$ (Stable OU process), $\theta=0$ (Wiener process), and $\theta<0$ (Unstable OU process). The value of $\beta$ is equal to the optimum value of the time-average expected estimation error. The computation of $\beta$ remains the same irrespective of the signal models.

\item Further, we consider noisy samples and obtain an explicit expression for $\mathsf{mse}$. From the expression, we establish a performance upper bound of $\mathsf{mse}$. 

\item Our results hold for general \emph{i.i.d.} transmission time distributions of the queueing server with a finite mean.

\end{itemize}


%% file: sec_related_work.tex
\subsection{Related Work}

The results in this paper are tightly connected to the area of remote estimation, e.g., \cite{Nuno2011,Hajek2008,Rabi2012,nayyar2013,Basar2014,GAO201857,ChakravortyTAC2020,Wiener_TIT,Ornee_TON,GuoISIT2020,TsaiINFOCOM2020,Arafa_2020}. Optimal sampling policy of Wiener processes with a zero channel delay was studied in \cite{Rabi2012,Basar2014}, whereas we consider random \emph{i.i.d.} channel delay. A discrete-time optimal stopping problem was solved by using Dynamic programming in \cite{Rabi2012} to find the optimal sampling policy of OU processes. In \cite{Ornee_TON}, an optimal sampler of stable OU processes is obtained analytically where the sampling is suspended when the server is busy and is reactivated once the server becomes idle. The optimal sampling policy for Wiener processes in \cite{Wiener_TIT} and stable OU processes in \cite{Ornee_TON} is a special case of ours. Remote estimation of Wiener processes with random two-way delay was considered in \cite{TsaiINFOCOM2020}.

In \cite{GuoISIT2020}, a jointly optimal sampler, quantizer, and estimator design were found for a class of continuous-time Markov processes under a bit-rate constraint. In \cite{Arafa_2020}, the quantization and coding schemes on the estimation performance are studied. We consider noisy channels with random delay to establish performance bounds. A recent survey on remote estimation systems was presented in \cite{jog2019channels}.


%% file: model.tex
\section{Model and Problem Formulation} \label{model} 

\begin{figure}
\vspace{-0.3cm}
\centering
\includegraphics[width=8.7cm]{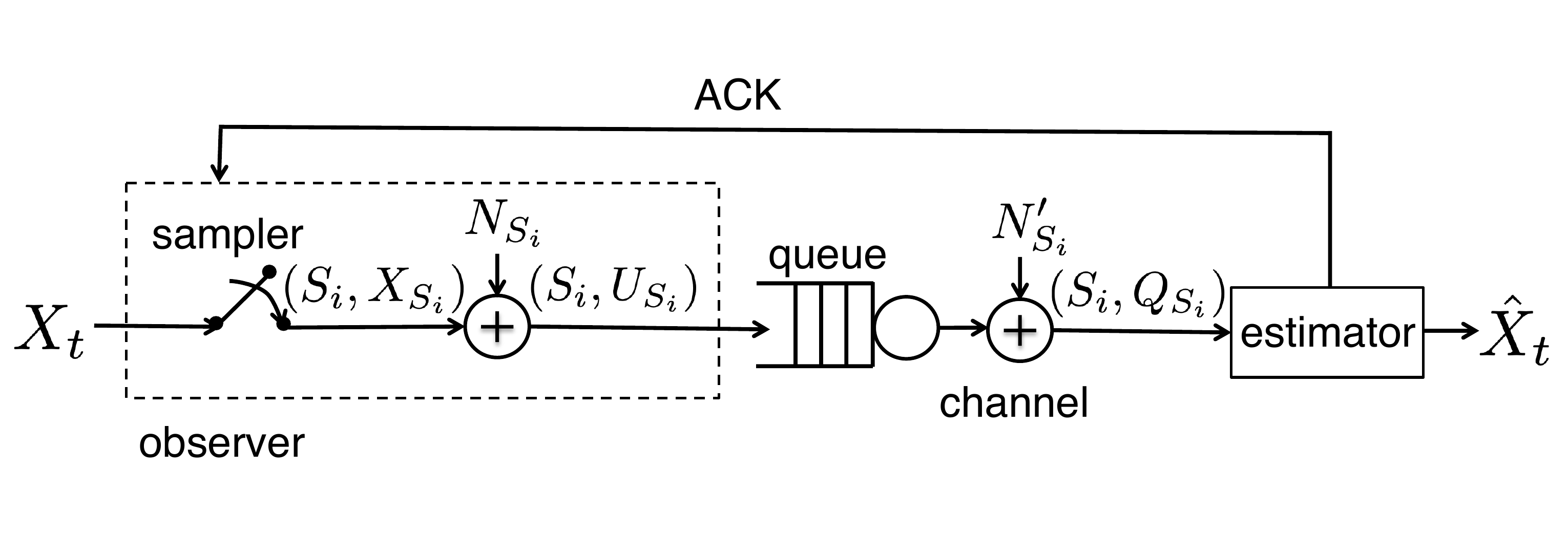}  
\caption{System model.}\vspace{-0.0cm}
\label{fig_model}
\end{figure}  

\subsection{System Model}
We consider a continuous-time remote estimation system that is illustrated in Fig. \ref{fig_model}, where an observer takes samples from an OU process $X_t$. After sampling, additional noises from the sampler and the channel are added to the samples. Then, the noisy samples are sent to the estimator. The channel is modeled as a single-server FIFO queue with \emph{i.i.d.} service times. The samples undergo random service times in the channel due to fading, interference, congestions, etc. We also consider that at a time, only one sample can be delivered through the channel.

The operation of the system starts at time instant $t=0$. The generation time of the $i$-th sample is $S_i$, which satisfy $S_i \leq S_{i+1}$ for all $i$. Then, $i$-th sample undergoes a random service time $Y_i$, and is delivered to the estimator at time $D_i$, where  $S_i +Y_i \leq D_i$, $D_i +Y_{i+1}\leq D_{i+1}$, and $0<\mathbb{E}[Y_i] < \infty$ hold for all $i$. The $i$-th sample packet $(S_i, X_{S_i})$ contains the sample value $X_{S_i}$ and its sampling time $S_i$. Suppose that after sampling, noise $N_{S_i}$ is being added to the sample $X_{S_i}$ and  the noisy observation of the sample $X_{S_i}$ is denoted by $U_{S_i}$. Hence, 
\begin{align} \label{quant_sample}
U_{S_i} = X_{S_i} + N_{S_i}, 
\end{align}
where $N_{S_i}$ is the additive noise with zero mean and variance $b_1$. Each sample packet $(S_i, U_{S_i})$ contains the sampling time $S_i$ and the noisy sample $U_{S_i}$. If channel noise $N'_{S_i}$ with zero mean and variance $b_2$ is added to the sample during its transmission through the channel, then the sample value becomes
\begin{align} \label{channel_noise}
Q_{S_i} = U_{S_i} + N'_{S_i}.
\end{align}

Initially, at $t=0$, the state of the system is assumed to hold $S_0 = 0$, and  $D_0 = Y_0$. 
The initial state of the OU process $X_0$ is a finite constant. The process parameters $\mu$, $\theta$, and $\sigma$ in \eqref{eq_SDE} are known at both the sampler and estimator. 

Let, the idle/busy state of the server at time $t$ is denoted by $I_t \in \{0,1\}$. We also assume that an acknowledgement is immediately sent back to the sampler whenever a sample is delivered and this operation has zero delay. By this assumption, the sampler is aware of the idle/busy state of the server and the available information at time $t$ can be given by $\{X_s, I_s: 0\leq s\leq t\}$.

\subsection{Sampling Policies}

The sampling time $S_i$ is a finite stopping time with respect to the filtration $\{{\mathcal{F}_{t}^{+}}, t \geq 0\}$ (a non-decreasing and right-continuous family of $\sigma$-fields) of the information that is available at the sampler such that \cite{Durrettbook10}
\begin{align} \label{stopping_time}
\{S_i \leq t\} \in \mathcal{F}_{t}^{+}, \forall t \geq 0.
\end{align}

Let $\pi = (S_1, S_2, . . . )$ denote a sampling policy and $\Pi$ denote the set of \emph{causal} sampling policies that satisfy two conditions: (i) Each sampling policy $\pi\in\Pi$ satisfies \eqref{stopping_time} for all $i$. (ii) The sequence of inter-sampling times $\{T_i = S_{i+1}-S_i, i=0,1,\ldots\}$ forms a \emph{regenerative process} \cite[Section IIB]{Ornee_TON}: An increasing sequence  $0\leq {l_1}<l_2< \ldots$ of  almost surely finite random integers exists such that the post-${l_k}$ process $\{T_{l_k+i}, i=0,1,\ldots\}$ is independent of the pre-$l_k$ process $\{T_{i}, i=0,1,\ldots, l_k-1\}$ and has same distribution as the post-${l_0}$ process $\{T_{l_0+i}, i=0,1,\ldots\}$; We further assume that $\mathbb{E}[{l_{k+1}}-{l_k}]<\infty$, $\mathbb{E}[S_{l_{1}}]<\infty$, and $0<\mathbb{E}[S_{l_{k+1}}-S_{l_k}]<\infty, ~k=1,2,\ldots$ 

\subsection{MMSE Estimator}

In this section, we provide the MMSE estimator for noisy samples of the OU process.


By using the expression of OU process for stable scenario \cite[Eq. (3)]{Maller2009} and the strong Markov property of the OU process \cite[Eq. (4.3.27)]{Peskir2006}, a solution to \eqref{eq_SDE} for $t \in [S_i,\infty)$ given by the following three cases:
\begin{align}\label{eq_solution}
\!\! \!\! X_t  = \left\{\!\! \begin{array}{l l} X_{S_i} e^{-\theta (t-S_i)} + \mu\big[1-e^{-\theta (t-S_i)} \big]\\ 
 + \frac{\sigma}{\sqrt{2\theta}}e^{-\theta (t-S_i)} W_{e^{2 \theta (t-S_i)}-1},
& \text{ if }~ \theta > 0,\\
{\sigma} W_t, & \text{ if }~ \theta = 0,\\
X_{S_i} e^{-\theta (t-S_i)} + \mu\big[1-e^{-\theta (t-S_i)} \big] \\
+ \frac{\sigma}{\sqrt{-2\theta}}e^{-\theta (t-S_i)} W_{1 - e^{2 \theta (t-S_i)}}, & \text{ if }~ \theta < 0.
\end{array}\right. \!\!\!\!\!\!
\end{align}

The estimator uses causally received samples to formulate an estimate $\hat X_t$ of the real-time signal value $X_t$ at any time $t\geq0$. The available information at the estimator has two parts: (i) $M_t = \{(S_i, Q_{S_i}, D_i): D_i \leq t\}$, which contains the sampling time $S_i$, noisy sample value $Q_{S_i}$, and delivery time $D_i$ of the samples that have been delivered by time $t$ and (ii) no sample has been received after the last  delivery time $\max\{D_i: D_i\leq t\}$. Similar to \cite{Rabi2012,SOLEYMANI20161,Wiener_TIT, Ornee_TON}, we assume that the estimator neglects the second part of information. Then, as shown in \cite{Arafa_2020}, the MMSE estimator for $t\in[D_i,D_{i+1}),~i=0,1,2,\ldots$ for all of the cases in \eqref{eq_solution} is given as follows
\begin{align} \label{mse_quant_u}
\!\! \!\! \hat{X}_{t}  = &\mathbb{E} [X_t | M_t] \nonumber\\
=& Q_{S_i} e^{-\theta (t-S_i)} + \mu\big[1-e^{-\theta (t-S_i)} \big].
\end{align}

\subsection{Performance Metric}

We evaluate the performance of remote estimation by the time-average mean square error which is expressed as follows:
\begin{align} \label{mse_new}
\mathsf{mse} = \limsup_{T\rightarrow \infty}\frac{1}{T}\mathbb{E}\left[\int_0^{T} (X_t - \hat X_t)^2dt\right].
\end{align}

A lower bound of \eqref{mse_new} can be obtained when the additive noises are not considered ($N_{S_i} = 0, N'_{S_i} = 0$). On the other hand, an upper bound can be found by taking both the noises into account. Moreover, we formulate the following optimal sampling problem that minimizes the time-average mean-squared estimation error over an infinite time-horizon when no noise is considered. 
\begin{align}\label{eq_DPExpected}
{\mathsf{mse}}_{\text{opt-wn}}=
\min_{\pi\in\Pi}~& \limsup_{T\rightarrow \infty}\frac{1}{T}\mathbb{E}\left[\int_0^{T} (X_t - \hat X_t)^2dt\right],
\end{align}
where ${\mathsf{mse}}_{\text{opt-wn}}$ is the optimum value of \eqref{eq_DPExpected} without noise.
We do not provide the optimal sampling policy in the presence of noises in this study, but it will be considered in our future work.

%% file: main_results.tex
\section{Main Results}\label{sec_main_result}

In this section, we first present the lower bounds for $\mathsf{mse}$ in \eqref{mse_new} for different conditions on the OU process parameter $\theta$. Second, we provide the optimal sampling policy for minimizing the expected estimation error defined in \eqref{eq_DPExpected}. Later, we present upper bound for $\mathsf{mse}$ in \eqref{mse_new}.


\subsection{Lower Bounds for \emph{$\mathsf{mse}$}}

Let us consider an OU process with initial state $O_0 = 0$ and parameter $\mu = 0$, which can be expressed as
\begin{align}\label{eq_signals}
\!\! \!\! O_t  = \left\{\!\! \begin{array}{l l}\frac{\sigma}{\sqrt{2 \theta}} e^{-\theta t} W_{e^{2 {\theta} t} - 1},
& \text{ if }~ \theta > 0,\\
{\sigma} W_t, & \text{ if }~ \theta = 0,\\
\frac{\sigma}{\sqrt{-2 \theta}} e^{-\theta t} W_{1-e^{2 {\theta} t}}, & \text{ if }~ \theta < 0.
\end{array}\right. \!\!\!\!\!\!
\end{align}
Before presenting the optimal sampler without noise, let us define the following parameter:
{\begin{align}\label{eq_mse_Yi}
&\mathsf{mse}_{Y_i}=\left\{\!\! \begin{array}{l l}\frac{\sigma^2}{2\theta} \mathbb{E} [1 - e^{{-2 \theta {Y_i}}}],
& \text{ if }~ \theta \neq 0,\\
{\sigma^2} \mathbb{E} [Y_i], & \text{ if }~ \theta = 0,
\end{array}\right. \!\!\!\!\!\!
\end{align}
where {$\mathsf{mse}_{Y_i}$} is the lower bound of $\mathsf{{mse}}$.
We will also need to use the following two functions 
\begin{align}
& G(x) = \thinspace \frac{\sqrt{\pi}}{2} \frac{e^{{x^2}}}{x}  \thinspace {\text{erf}}(x), ~x\in[0,\infty), \label{eq_g2}\\
& K(x) = \thinspace \frac{\sqrt{\pi}}{2} \frac{e^{-{x^2}}}{x}  \thinspace {\text{erfi}}(x), ~x\in[0,\infty), \label{K_x}
\end{align}
where if $x=0$, both $G(x)$ and $K(x)$ are defined as their right limits $G(0)=\lim_{x\rightarrow 0^+}G(x)= 1$, and $K(0)=\lim_{x\rightarrow 0^+}K(x)= 1$. Furthermore, $\text{erf}(\cdot)$ and $\text{erfi}(\cdot)$ are the error function and  imaginary error function respectively, defined as
\begin{align}
{\text{erf}}(x) = \frac{2}{\sqrt \pi} \int_0^x e^{-t^2} dt, 
{\text{erfi}}(x) = \frac{2}{\sqrt \pi} \int_0^x e^{t^2} dt. \label{eq_erfi}
\end{align}
Note that $G(x)$ is strictly increasing on $x\in[0, \infty)$ \cite{Ornee_TON}, whereas $K(x)$ is strictly decreasing on $x\in[0, \infty)$. Hence, their inverses $G^{-1} (\cdot)$ and $K^{-1} (\cdot)$  are properly defined.

First, we consider that the system has no noise, i.e., $N_{S_i} = 0$ and $N'_{S_i} = 0$. Therefore, from \eqref{quant_sample} and \eqref{channel_noise}, we get, $X_{S_i} = U_{S_i} = Q_{S_i}$. Then, the following theorem illustrates that the optimal sampling policy is a threshold policy and the threshold is found for all the three cases of the OU process parameter $\theta$. 

\begin{theorem}\label{thm1}
If the $Y_i$'s are i.i.d. with $0<\mathbb{E}[Y_i] < \infty$, then $(S_1(\beta),S_2(\beta),\ldots)$  with a parameter $\beta$ is an optimal solution to \eqref{eq_DPExpected}, where 
\begin{align}\label{eq_opt_solution}
S_{i+1} (\beta)= \inf \left\{ t \geq D_i(\beta):\! \big|X_t - \hat X_t\big| \!\geq\! {v}(\beta)\right\},
\end{align}
$D_i (\beta)= S_i (\beta)+ Y_i$, and ${v}(\beta)$ is given by
\begin{align}\label{eq_threshold_st}
\!\! \!\! v(\beta)  = \left\{\!\! \begin{array}{l l}\frac{\sigma}{\sqrt{\theta}} G^{-1} \left(\frac{\frac{\sigma^2}{2 \theta} - \mathsf{mse}_{Y_i}}{\frac{\sigma^2}{2 \theta} - \beta}\right),
& \text{ if }~ \theta > 0,\\
\sqrt{3(\beta - \mathbb{E} [Y_i])}, & \text{ if }~ \theta = 0,\\
\frac{\sigma}{\sqrt{-\theta}} K^{-1} \left(\frac{\frac{\sigma^2}{2 \theta} - \mathsf{mse}_{Y_i}}{\frac{\sigma^2}{2 \theta} - \beta}\right), & \text{ if }~ \theta < 0,
\end{array}\right. \!\!\!\!\!\!
\end{align}
where $G^{-1}(\cdot)$ is the inverse function of $G(\cdot)$ in \eqref{eq_g2}, $K^{-1}(\cdot)$ is the inverse function of $K(\cdot)$ in \eqref{K_x}, and $\beta$ is the unique root  of
\begin{align}\label{thm1_eq22}
\!\! {\mathbb{E}\left[\int_{D_i(\beta)}^{D_{i+1}(\beta)}\! (X_t-\hat X_t)^2dt\right]} - {\beta} {\mathbb{E}[D_{i+1}(\beta)\!-\!D_i(\beta)]} = 0.\!\!
\end{align} 
The optimal objective value to \eqref{eq_DPExpected} is then given by \emph{
\begin{align}\label{thm1_eq23}
{\mathsf{mse}}_{\text{opt-wn}} = \frac{\mathbb{E}\left[\int_{D_i(\beta)}^{D_{i+1}(\beta)}\! (X_t-\hat X_t)^2dt\right]}{\mathbb{E}[D_{i+1}(\beta)\!-\!D_i(\beta)]}.
\end{align} }
\end{theorem}

In \cite{Ornee_TON}, it is proved that the optimal sampling policy for stable OU process, i.e., when $\theta > 0$ is a threshold policy. The threshold obtained in \cite{Ornee_TON} coincides with $v(\beta)$ in \eqref{eq_threshold_st} for the case of $\theta > 0$. For $\theta = 0$, the threshold is obtained for $\sigma = 1$ which represents a Wiener process \cite{Wiener_TIT}. For $\theta<0$, the proof procedure works in the same way as explained in \cite{Ornee_TON} for stable OU processes. The threshold $v(\beta)$ is obtained by solving similar free boundary problems explained in \cite{Ornee_TON} and the optimality of \eqref{thm1_eq23} for $\theta < 0$ is thus guaranteed. However, the threshold structure is different for all the three cases in Theorem \ref{thm1}. The function $K(x)$ in \eqref{K_x} is related to the function $G(x)$ in \eqref{eq_g2} as follows
\begin{align} \label{GK_relation}
K(x) = G(jx),
\end{align}
where $j$ is the imaginary number represented by $j=\sqrt{-1}$.
Therefore, the threshold $v(\beta)$ for $\theta < 0$ can be expressed by the following equation as well:
\begin{align}
v(\beta) = j^{-1} \frac{\sigma}{\sqrt{-\theta}} G^{-1} \left(\frac{\frac{\sigma^2}{2 \theta} - \mathsf{mse}_{Y_i}}{\frac{\sigma^2}{2 \theta} - \beta}\right).
\end{align}


Though the threshold functions $v(\beta)$ varies with signal structure, the computation of the parameter $\beta$ remains the same for all cases and the uniqueness of the root of \eqref{thm1_eq22} is proved in \cite{Ornee_TON}. The decision of taking a new sample defined in \eqref{eq_opt_solution} works in the same way as explained in \cite{Ornee_TON}. 

\subsection{Upper Bounds for \emph{$\mathsf{mse}$}}

Suppose that the additive noise in the sampler and channel exist in the system, i.e., $N_{S_i} \neq 0, N'_{S_i} \neq 0$. Moreover, the sampler follows the sampling strategy obtained in \eqref{eq_opt_solution}. The OU process $O_t$ is a Gauss-Markov process. When noises get incorporated with $O_t$, it does not remain Markov. The analysis presented in our previous study \cite{Ornee_TON} was based on the strong Markov property of the OU processes. Due to the non-Markovian structure of noisy samples of OU processes, finding an optimal sampling policy requires different analytical tools. Due to lack of space, we do not provide the optimal sampling policy with the presence of noise, but it will be considered in our future study. 

Because the noises $N_{S_i}$ and $N'_{S_i}$ are independent of the sampling times and the observed OU process, by utilizing \eqref{quant_sample}, \eqref{channel_noise}, \eqref{eq_solution}, \eqref{mse_quant_u}, and \eqref{thm1_eq23}, the $\mathsf{mse}$ at the estimator which is an upper bound of \eqref{mse_new} can be expressed as
\begin{align} \label{mse_quant1}
 \mathsf{mse} = & \frac{\mathbb{E}\left[\int_{D_i(\beta)}^{D_{i+1}(\beta)}\! (X_t-\hat X_t)^2dt\right]}{\mathbb{E}[D_{i+1}(\beta)\!-\!D_i(\beta)]} \nonumber\\
 =& \frac{\mathbb{E}\left[\int_{D_i(\beta)}^{D_{i+1}(\beta)}\! (O_{t-{S_i}} - (N_{S_i} + N'_{S_i}) e^{-\theta (t-S_i)})^2 dt\right]}{\mathbb{E}[D_{i+1}(\beta)\!-\!D_i(\beta)]} \nonumber\\
 =& \frac{\mathbb{E}\left[\int_{D_i(\beta)}^{D_{i+1}(\beta)}\! O^{2}_{t-{S_i}} dt\right]}{\mathbb{E}[D_{i+1}(\beta)\!-\!D_i(\beta)]} \nonumber\\
& +\frac{\mathbb{E}\left[\int_{D_i(\beta)}^{D_{i+1}(\beta)}\! (N_{S_i} + N'_{S_i})^2 e^{-2 \theta (t-S_i)} dt\right]}{\mathbb{E}[D_{i+1}(\beta)\!-\!D_i(\beta)]}, 
 \end{align}
 where \eqref{mse_quant1} follows due to the fact that the OU process $O_t$ has initial state $O_0 = 0$ and the noises $N_{S_i}$ and $N'_{S_i}$ with zero mean are independent of the observed OU process and sampling times.
 
To compute \eqref{mse_quant1}, the first fractional term remains the same as the $\mathsf{mse}_{\text{opt-wn}}$ in \eqref{thm1_eq23} with $N_{S_i} = 0$ and $N'_{S_i}=0$. For stable OU processes, the associated $\mathsf{mse}_{\text{opt-wn}}$ is computed in \cite[Lemma 1] {Ornee_TON}. The expression of $O^2_{t-Si}$ is the same for both stable and unstable OU processes. Therefore, the solution for \eqref{mse_quant1} holds for all three cases in \eqref{eq_signals}. For computing the second term, as $N_{S_i}$ and $N'_{S_i}$ are independent of the observed OU process and the sampling times, the numerator of the second fractional term in \eqref{mse_quant1} can be written as:
\begin{align}
& {\mathbb{E}\left[\int_{D_i(\beta)}^{D_{i+1}(\beta)}\! (N_{S_i} + N'_{S_i})^2 e^{-2 \theta (t-S_i)} dt\right]} \nonumber\\
=& \mathbb{E} [(N_{S_i} + N'_{S_i})^2] {\mathbb{E}\left[\int_{D_i(\beta)}^{D_{i+1}(\beta)}e^{-2 \theta (t-S_i)} dt\right]}. \label{new_lemma}
\end{align}
Then, we have the following lemma for the last term in \eqref{new_lemma}.

\begin{lemma} \label{lemma_W(v)}
It holds that
\begin{align} \label{W_v}
& {\mathbb{E}\left[\int_{D_i(\beta)}^{D_{i+1}(\beta)} e^{-2 \theta (t-S_i)} dt\right]} \nonumber\\
=& \!\! \frac{1}{2 \theta} \mathbb{E} \!\! \bigg[\! e^{-2 \theta Y_i} \bigg\{1 - {\min} \bigg(1, \frac{{} {_1F_1} \big(1, \frac{1}{2}, \frac{\theta}{\sigma^2} O^{2}_{Y_i}\big)}{{} {_1F_1} \big(1, \frac{1}{2}, \frac{\theta}{\sigma^2} v^{2} (\beta) \big)}\bigg) \mathbb{E} [e^{-2 \theta Y_{i+1}}] \bigg\} \!\! \bigg]. \!\!
\end{align}
\end{lemma}

\begin{proof}
See Appendix \ref{lemma_proof}.
\end{proof}

By using Lemma \ref{lemma_W(v)} and the expressions obtained in \cite[Lemma 1]{Ornee_TON}, all the associated expectations in \eqref{mse_quant1} can be obtained by Monte Carlo simulations of scalar random variables $O_{Y_i}$ and $Y_i$, which does not require to directly simulate the entire random process $\{O_t, t \geq 0\}$.

%% file: numerical_results.tex
\section{Numerical Results}\label{sec_num}

\begin{figure}
	\centering
	\includegraphics[width=0.5\textwidth]{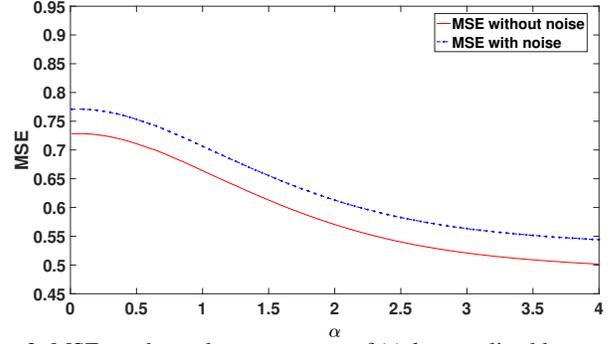}   
	\caption{MSE vs. the scale parameter ${\alpha}$ of \emph{i.i.d.} normalized log-normal service time distribution with $\mathbb{E}[Y_i]=1$, where the parameters of the OU process are $\sigma=1$ and $\theta=0.5$.}
	\label{noise_lognormal}
	\vspace{-3mm}
\end{figure}

Figure \ref{noise_lognormal} illustrates the MSE of \emph{i.i.d} normalized log-normal service time, where $Y_i = {e^{{\alpha} {X_i}}}/{\mathbb{E}[e^{{\alpha} {X_i}}]},$ and ${\alpha}>0$ is the scale parameter of log-normal distribution. The $(X_1, X_2, \dots)$ are \emph{i.i.d.} Gaussian random variables, where $\mathbb{E} [X_i] = 0$ and $\text{Var} (X_i) = 1$. The maximum throughput of the queue is 1 as $\mathbb{E}[Y_i]=1$. Both of the noises $N_{S_i}$ and $N'_{S_i}$ are considered to have 0 mean and variance 0.1.
With the growth of the scale parameter $\alpha$, the tail of the log-normal distribution becomes heavier. The MSE with noise curve shows performance degradation as the additional term due to noise added with the $\mathsf{mse}$ without noise.

%% file: sec_conclusion.tex
\section{Conclusion} \label{conclusion}
In this paper, we have explained the optimal sampling strategies for minimizing the \emph{instantaneous estimation error} for three different cases of scalar Gauss-Markov signal processes. The optimal sampler exhibits a threshold policy and by using causal knowledge of the signal values, a smaller estimation error has been obtained. The optimal threshold has been changed with signal structure. For noisy samples, the additional term added in the $\mathsf{mse}$ due to noise is found. An optimal sampler design for noisy samples of Gauss-Markov processes will be considered in our future study.

%% file: sec_appendix.tex
\section{Proof of Lemma \ref{lemma_W(v)}} \label{lemma_proof}

In order to prove Lemma \ref{lemma_W(v)}, we need to consider the following two cases:

\emph{Case 1:} If $|X_{D_i (\beta)} - \hat X_{D_i (\beta)}| = |O_{Y_i}| \geq v(\beta)$, then $S_{i+1} (\beta) = D_i (\beta)$. Hence, 
\begin{align}\label{eq_expectation_12}
D_{i} (\beta) &= S_{i}(\beta) + Y_{i}, \\
D_{i+1} (\beta)& = S_{i+1}(\beta) + Y_{i+1} = D_i (\beta) + Y_{i+1}.
\end{align}
Let us consider the following equation:
\begin{align}
& {\mathbb{E}\left[\int_{D_i(\beta)}^{D_{i+1}(\beta)} e^{-2 \theta (t-S_i)} dt \Big| O_{Y_i} = q, Y_i =y, |O_{Y_i}| \geq v(\beta)
\right]} \nonumber\\
=& \mathbb{E} \left[\int_{Y_i}^{Y_i + Y_{i+1}} e^{-2 \theta s} ds \Big| O_{Y_i} = q, Y_i =y, |O_{Y_i}| \geq v(\beta)
\right] \nonumber\\
=& \mathbb{E} \bigg[ \frac{1}{2 \theta} e^{-2 \theta y} (1 - e^{-2 \theta Y_{i+1}}) \big| O_{Y_i} = q, Y_i =y, |O_{Y_i}| \geq v(\beta) \bigg] \nonumber\\
=& \frac{1}{2 \theta} e^{-2 \theta y} \mathbb{E} \bigg[1 - e^{-2 \theta Y_{i+1}} \big| O_{Y_i} = q, Y_i =y, |O_{Y_i}| \geq v(\beta) \bigg] \nonumber\\
=& \frac{1}{2 \theta} e^{-2 \theta y} \big\{1 - \mathbb{E} [e^{-2 \theta Y_{i+1}}]\big\}, \label{lem1_eq1}
\end{align}
where \eqref{lem1_eq1} holds due to the fact that $Y_{i+1}$ is independent of $O_{Y_i}$ and $Y_i$.
\\

\emph{Case 2:} If $|X_{D_i(\beta)} - \hat X_{D_i(\beta)} |= |O_{Y_i}| < v(\beta)$, then 
\begin{align}
& {\mathbb{E}\left[\int_{D_i(\beta)}^{D_{i+1}(\beta)} e^{-2 \theta (t-S_i)} dt \Big| O_{Y_i} = q, Y_i =y, |O_{Y_i}| < v(\beta)
\right]} \nonumber\\
=& \mathbb{E} \left[\int_{Y_i}^{Y_i + Z_i+ Y_{i+1}} e^{-2 \theta s} ds \Big| O_{Y_i} = q, Y_i =y, |O_{Y_i}| < v(\beta)
\right] \nonumber\\
=& \mathbb{E} \bigg[ \frac{1}{2 \theta} e^{-2 \theta y} (1 - e^{-2 \theta Z_i} e^{-2 \theta Y_{i+1}}) 
\big| O_{Y_i} = q, Y_i =y \bigg] \nonumber\\
=& \frac{1}{2 \theta} e^{-2 \theta y} \mathbb{E} \bigg[1 - e^{-2 \theta Z_i} e^{-2 \theta Y_{i+1}} \big| O_{Y_i} = q, Y_i =y \bigg] \nonumber\\
=& \frac{1}{2 \theta} e^{-2 \theta y} \bigg\{1 - \mathbb{E} \bigg[e^{-2 \theta Z_i} \Big| O_{Y_i} = q, Y_i =y \bigg] \mathbb{E} [e^{-2 \theta Y_{i+1}}]\bigg\} \nonumber\\
=& \frac{1}{2 \theta} e^{-2 \theta y} \bigg\{1 - \mathbb{E} [e^{-2 \theta Z_i} | O_{Y_i} = q] \mathbb{E}[ e^{-2 \theta Y_{i+1}}] \bigg\}, \label{lem1_eq2} 
\end{align}
where the last equation in \eqref{lem1_eq2} holds because $Z_i$ is conditionally independent of $Y_i$ given $O_{Y_i}$. Next, we need to compute $\mathbb{E} [e^{-2 \theta Z_i} | O_{Y_i} = q]$, where $Z_i$ is a hitting time of the time-shifted OU process $O_{t+Y_i}$ given as
\begin{align}
& Z_i = \nonumber\\
& \text{inf} \{t : O_{t + Y_i} \not\in (-v(\beta), v(\beta)) | O_{Y_i} = q \in (-v(\beta), v(\beta))\}.
\end{align}
By using the characteristic function of the hitting time of the OU process in \cite[Eq. 15a] {Khem2019}, we get that
\begin{align}
\mathbb{E} [e^{-2 \theta Z_i} | O_{Y_i} = q] = \frac{{} {_1F_1} \big(1, \frac{1}{2}, \frac{\theta}{\sigma^2} q^2 \big)}{{} {_1F_1} \big(1, \frac{1}{2}, \frac{\theta}{\sigma^2} v^{2} (\beta) \big)}.
\end{align}
Therefore, \eqref{lem1_eq2} becomes
\begin{align}
& {\mathbb{E}\left[\int_{D_i(\beta)}^{D_{i+1}(\beta)} e^{-2 \theta (t-S_i)} dt \Big| O_{Y_i} = q, Y_i =y, |O_{Y_i}| < v(\beta)
\right]} \nonumber\\
=& \frac{1}{2 \theta} e^{-2 \theta y} \bigg\{1 - \frac{{} {_1F_1} \big(1, \frac{1}{2}, \frac{\theta}{\sigma^2} q^2 \big)}{{} {_1F_1} \big(1, \frac{1}{2}, \frac{\theta}{\sigma^2} v^{2} (\beta) \big)} \mathbb{E}[ e^{-2 \theta Y_{i+1}}] \bigg\}. \label{lem1_eq3}
\end{align}
By combining \eqref{lem1_eq1} and \eqref{lem1_eq3}, we get that
\begin{align}
& {\mathbb{E}\left[\int_{D_i(\beta)}^{D_{i+1}(\beta)} e^{-2 \theta (t-S_i)} dt \Big| O_{Y_i} = q, Y_i =y
\right]} \nonumber\\
=& \frac{1}{2 \theta} e^{-2 \theta y} \bigg[ 1 - {\min} \bigg\{ 1, \frac{{} {_1F_1} \big(1, \frac{1}{2}, \frac{\theta}{\sigma^2} q^2 \big)}{{} {_1F_1} \big(1, \frac{1}{2}, \frac{\theta}{\sigma^2} v^{2} (\beta) \big)} \bigg\} \mathbb{E} [e^{-2 \theta Y_{i+1}}] \bigg]. \label{lem1_eq4}
\end{align}
Finally, by taking the expectation over $O_{Y_i}$ and $Y_i$ in \eqref{lem1_eq4}, Lemma \ref{lemma_W(v)} is proven.